\documentclass[aps,pra,twocolumn,groupedaddress]{revtex4-1}
\usepackage{graphicx}
\usepackage{epstopdf}
\usepackage{braket}
\usepackage{lipsum}

\usepackage{amsmath,amssymb,amsfonts,amsthm}
\usepackage{multirow}
\usepackage{verbatim}
\usepackage{url}
\usepackage{comment}
\usepackage{mathtools}
\usepackage{epsf,pgf,graphicx}
\usepackage{color}
\usepackage{tikz,pgf}
\usepackage{makecell}
\newtheorem{proposition}{Proposition}
\newtheorem{lemma}{Lemma}
\newtheorem{theorem}{Theorem}
\newtheorem{corollary}{Corollary}

\newtheorem{construction}{Construction}
\newtheorem{remark}{Remark}
\newtheorem{example}{Example}

\DeclarePairedDelimiter\ceil{\lceil}{\rceil}
\DeclarePairedDelimiter\floor{\lfloor}{\rfloor}

\usepackage{adjustbox}
\usepackage[colorlinks = true,
            linkcolor = blue,
            urlcolor  = blue,
            citecolor = blue,
            anchorcolor = blue]{hyperref}
\usepackage{subcaption}

\begin{document}
\title{On Exact Space-Depth Trade-Offs in Multi-Controlled Toffoli Decomposition}
\author{Suman Dutta$^{1,3}$, Siyi Wang${^1}$, Anubhab Baksi${^2}$, Anupam Chattopadhyay${^1}$, Subhamoy Maitra${^3}$}
\affiliation{$^1$College of Computing \& Data Science, Nanyang Technological University, Singapore, 
$^2$School of Physical \& Mathematical Sciences, Nanyang Technological University, Singapore,
$^3$Applied Statistics Unit, Indian Statistical Institute, Kolkata, India}
\email{sumand.iiserb@gmail.com\\ siyi002@e.ntu.edu.sg\\ anubhab001@e.ntu.edu.sg\\ anupam@ntu.edu.sg\\ subho@isical.ac.in}
\begin{abstract}
In this paper, we consider the optimized implementation of Multi Controlled Toffoli (MCT) using the Clifford $+$ T gate sets. While there are several recent works in this direction, here we explicitly quantify the trade-off (with concrete formulae) between the Toffoli depth (this means the depth using the classical 2-controlled Toffoli) of the $n$-controlled Toffoli (hereform we will tell $n$-MCT) and the number of clean ancilla qubits. Additionally, we achieve a reduced Toffoli depth (and consequently, T-depth), which is an extension of the technique introduced by Khattar et al. (2024). In terms of a negative result, we first show that using such conditionally clean ancillae techniques, Toffoli depth can never achieve exactly $\ceil{\log_2 n}$, though it remains of the same order. This highlights the limitation of the techniques exploiting conditionally clean ancillae [Nie et al., 2024, Khattar et al., 2024]. Then we prove that, in a more general setup, the T-Depth in the Clifford + T decomposition, via Toffoli gates, is lower bounded by $\ceil{\log_2 n}$, and this bound is achieved following the complete binary tree structure. Since the ($2$-controlled) Toffoli gate can further be decomposed using Clifford $+$ T, various methodologies are explored too in this regard for trade-off related implications.
\end{abstract}
\maketitle

\section{Introduction}
\label{sec:intro}
Quantum gates are the fundamental building blocks of quantum circuits. Unlike classical gates, quantum gates are inherently reversible and are mathematically represented by unitary matrices. The doubly-controlled X-gate, commonly known as the Toffoli gate, is among the most significant quantum gates, with critical applications in arithmetic operations~\cite{wang2024boosting,wang2024optimal, wang2024Minimum, wang2023higher}, reversible computing~\cite{lighter2019, dorcis2023}, and oracle constructions~\cite{aes16, chacha21, fsr21, zuc24}. However, the Toffoli gate's high resource demands, particularly in terms of T-Count, T-Depth, and ancilla qubits, can greatly influence the efficiency of fault-tolerant quantum circuits. Consequently, optimizing its implementation is crucial for reducing computational overhead, minimizing error rates, and enhancing scalability, thereby making it indispensable for practical large-scale quantum computations.

An $n$-controlled Toffoli is an $(n+1)$-qubit gate, where $n$ many qubits control the outcome of a single target qubit.
Like ($2$-controlled) Toffoli, Multi Controlled Toffoli (MCT) also has a huge application in quantum arithmetic, error correction, and implementation of quantum algorithms. For example, in Grover's Search, implementing the $n$-bit AND function requires an $n$-controlled Toffoli ($n$-MCT) gate. As MCT cannot be implemented in its original form (so far), it needs to be decomposed into Clifford + Toffoli, and eventually to the Clifford + T gate set. Throughout the paper, By Toffoli, we refer to $n$-controlled Toffoli with $n=2$.

It is well known that decomposing/ designing larger circuits with smaller components requires additional qubits, known as ancilla qubits, to reduce depth. The important task of optimizing the number of gates, the depth, and the number of additional qubits has received serious attention for a long time, and one may refer to the seminal work of Moore and Nilsson~\cite{moore01} two decades back. A more recent and comprehensive discussion is available in~\cite{jiang20}, and we like to quote from that work, ``Can we characterize the relationship between the number of ancilla and the possible optimal depth?'' The paper~\cite{jiang20} is concerned with CNOT, whereas in this paper we concentrate on Toffoli. However, the fundamental motivation of quantum circuit synthesis revolves around this very question.

Over the past decades, numerous efforts have been made to reduce the resource requirements for the optimized implementation of multi-controlled Toffoli (MCT) gates \cite{miller11,saeedi13,maslov16,baker19,paler22,bala22,claudon24}. More recent works \cite{nie24,khattar24} have introduced the conditionally clean ancillae technique, which significantly reduces both the Toffoli depth and ancilla count in the MCT decomposition.

Before proceeding further, let us explain the concept of conditionally clean ancillae that will be repeatedly referred to in this work. Generally, ancilla qubits are of two types, namely the clean ancilla, which is initialized to $\ket{0}$ at the beginning of the circuit, and the dirty ancilla, which is initialized to some unknown quantum state at the beginning. The standard practice is to re-initialize the ancilla qubit to its initial state at the end of the computation. Conditionally clean ancillae \cite{nie24,khattar24} are the set of working (control) qubits, derived from the existing ones, that are treated as the clean ancilla based on certain assumptions, and re-initialized at the end of the computation. Although earlier works~\cite{nie24,khattar24} did consider the dirty ancilla, in this paper, we are focusing on the conditionally clean ancillae derived from the clean ones only. Very recently, the conditionally clean ancillae technique has gained significant attention, including its applications in quantum adders \cite{remaud25}.

The resource optimization metrics under consideration are the Toffoli count, Toffoli depth, and Ancilla count. The Toffoli count and Toffoli depth are further refined with the T-Count and T-Depth as one may refer to Table~\ref{tab:tof}.

In this paper, we explore how the Toffoli depth (and consequently, T-Depth) can be reduced by increasing the number of clean ancilla qubits, utilizing the concept of conditionally clean ancillae. We also establish the limitation of this method in terms of the lower bound of the Toffoli depth. Additionally, we show that, in a more general setting, the Toffoli depth (which can be further reduced to T-Depth) of an $n$-MCT decomposition cannot be reduced beyond $\ceil{\log_2 n}$. The section-wise contributions of this paper are outlined as follows.

\subsection{Organization and Contributions}
In Section \ref{sec:pre}, we proceed with the preliminaries.
Section \ref{sec:warmup} is a warm-up section where we first present the existing results on the Clifford $+$ T decomposition of ($2$-controlled) Toffoli and summarize the results in terms of T-Count, T-Depth, and ancilla requirements in Table \ref{tab:tof}. Additionally, we also explain the recent developments of MCT decompositions describing the existing best results in terms of Toffoli count, Toffoli depth, and ancilla, which can be subsequently decomposed into T-Count and T-Depth, as summarized in Table \ref{tab:mct}. Given that there are several developments in very recent times, this section provides a holistic view of the existing results. Based on this, we present our contributions.

In section \ref{sec:cont1}, we explore the trade-off between the (clean) ancilla and the Toffoli depth using the existing techniques related to conditionally clean ones. In Section \ref{sub:nvp}, we take a different look at viewing the MCT circuit decomposition using the conditionally clean ancillae technique of \cite{khattar24} and enumerate their exact Toffoli count. Following the trade-off, we show the reduction in Toffoli depth, and therefore in T-Depth (Construction \ref{cons:tradeoff}, Example \ref{ex:1}) compared to the recent work of Khattar and Gidney \cite{khattar24}, by introducing additional clean ancillae into the circuit, while keeping the Toffoli count constant. These results are shown in Section \ref{sub:tradeoff}. Then, in Section \ref{sub:lowerbound}, we identify the limitation of this technique in Theorem \ref{th:logn}, showing that this direction cannot reduce the Toffoli depth to $\ceil{\log_2 n}$, though it is of order $\mathcal{O}(\log_2 n)$.

In Section \ref{sec:cont2}, we prove within a general framework that the exact Toffoli depth in the ($2$-controlled) Toffoli decomposition of an $n$-MCT cannot be less than $\ceil{\log_2 n}$, regardless of the number of ancilla qubits used.
In fact, this lower bound can be achieved for T-Depth as well, by the construction of \cite{jaques19}. That is, using the technique of \cite{jaques19}, we can obtain an $n$-MCT by further decomposing into the Clifford + T gate set with an exact T-Depth of $\ceil{\log_2 n}$ too, using $2n-2$ ancillae, and a T-Count of $4(n-1)$. Moreover, using the logical-AND circuit by Gidney \cite{gidney18}, the ancilla count can be reduced to $n-2$ keeping the T-Count constant. However, the exact T-Depth in the case of \cite{gidney18} becomes one more, i.e., $\ceil{\log_2 n} +1$. To highlight our contribution, we are looking at the exact counts and depth instead of their complexity order.

Section \ref{sec:con} concludes the paper with a brief summary of our work and outlines the open problems in the related domain.

\section{Preliminaries}
\label{sec:pre}
In this section, we briefly proceed with the preliminaries.
A qubit is the fundamental unit of quantum information represented as $\ket{\psi} = \alpha \ket{0} + \beta\ket{1}$, where $\alpha, \beta \in\mathbb{C}$ and $|\alpha|^2 + |\beta|^2 = 1$. The basis states can also be written as column matrices, as follows.
$$\ket{0}=\begin{psmallmatrix}
     1\\ 0
\end{psmallmatrix}, \ket{1}=\begin{psmallmatrix}
    0\\ 1
\end{psmallmatrix}, \ket{\psi} = \alpha \ket{0} + \beta\ket{1} =\begin{psmallmatrix}
\alpha\\ \beta
\end{psmallmatrix}.$$
Upon measurement, the qubit collapses to one of the basis states, $\ket{0}$ with probability $|\alpha|^2$ or $\ket{1}$ with probability $|\beta|^2$. Similarly, an $n$-qubit state is described by $2^n$ parameters as $\ket{\psi_n}= \sum_{{x}\in\{0,1\}^n}\alpha_{{x}}\ket{{x}}$ with the normalization condition $\sum_{x\in\{0,1\}^n}|\alpha_{{x}}|^2 = 1$.

The Quantum gates, unlike the classical ones, are inherently reversible and represented by unitary matrices $U$, with inverses denoted by $U^{\dagger}$.
$$U^{\dagger}\left(U\ket{\psi} \right) = U\left(U^{\dagger}\ket{\psi} \right) = \ket{\psi}.$$

In classical computing, (2-input 1-output) NAND and NOR gates are considered universal as they can be used to construct any classical logic circuit. In contrast, quantum computing involves infinitely many quantum gates, including both single-qubit and multi-qubit ones, making it challenging to define a universal description. However, there exist certain gate sets that can approximate any unitary transformation on a quantum computer to an arbitrary degree of accuracy, known as the universal gate sets. The most common examples of quantum universal gate sets include the Clifford + T gate set, rotation gates combined with the CNOT gate set, etc.

The Clifford group is generated by three gates: Hadamard (H), phase (S), and CNOT. This set is minimal, as removing any one gate would result in the inability to implement some Clifford operations. Since all the Pauli matrices can be derived from the phase and Hadamard gates (Equation \ref{eq:pauli}), each Pauli gate is also an element of the Clifford group.
\begin{equation}
    \text{H } = \frac{1}{\sqrt{2}}\begin{pmatrix}
        1 & 1\\ 1 & -1
    \end{pmatrix},\,
    \text{S } = \begin{pmatrix}
        1 & 0\\ 0 & i
    \end{pmatrix},\,
    \text{CNOT } = \begin{pmatrix}
        I & 0\\ 0 & X
    \end{pmatrix},
\end{equation}
\begin{equation}
\label{eq:pauli}
    \text{I} = \text{H}^2,\, \text{X } = \text{HZH},\, \text{Y }= \text{S}^{\dagger}\text{XS},\, \text{Z } = \text{S}^2 = \text{HXH}.
\end{equation}
However, the Clifford gates alone do not constitute a universal set of quantum gates because certain gates, for example, the T $=\sqrt{\text{S}}$ gate, cannot be arbitrarily approximated using only Clifford operations. Therefore, the Clifford group, when augmented with the T gate, forms a universal quantum gate set for quantum computation.

The Toffoli gate is a three-qubit quantum gate, where the first two serve as control qubits, and the third one is the target. The gate flips the target qubit if and only if both control qubits are in the $\ket{1}$ state, described as $$\text{Toffoli }:\ket{x,y,z} \rightarrow \ket{x,y,z \oplus xy}.$$

Consequently, the Toffoli gate is also referred to as the doubly controlled-NOT gate or the CCNOT (CCX) gate. The other doubly controlled Pauli gates, such as the CCZ and CCY gates, are defined in a similar manner, satisfying: $\text{CCX } = \text{CC(HZH)} = \text{CC(SY}\text{S}^{\dagger}\text{)}$. The doubly controlled Pauli gates can be further extended to multi-controlled Pauli gates, and the corresponding equivalence is shown in Figure \ref{fig:mcp}.
\begin{figure}
    \centering
    \includegraphics[width=0.6\linewidth]{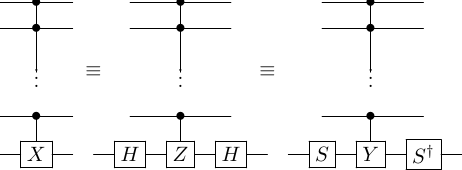}
    \caption{Equivalence between multi-controlled Pauli gates.}
    \label{fig:mcp}
\end{figure}
In quantum computing, the Multi-Controlled Toffoli (MCT) gates play a crucial role in the design of complex quantum algorithms, including error correction codes and arithmetic operations. Despite its utility, the implementation of MCT gates exploiting the Clifford + T gate set requires careful decomposition to optimize resource usage, such as minimizing the T-Count, T-Depth, and the number of ancilla qubits. These optimizations are essential for practical quantum computation, where resource efficiency is a critical consideration. In this regard, we present a comprehensive overview of the existing benchmarks for Clifford + T decompositions of single and multi-controlled Toffoli gates in the following section.

Let us now describe the idea of conditionally clean ancillae from~\cite{khattar24}, as described in Figure \ref{fig:ccanc}. Given a clean ancilla, a Toffoli gate is implemented targeting the clean ancilla. If the clean ancilla is reversed, it means both the control qubits were $\ket{1}$; thus, applying X gates on the control qubits will change them to $\ket{0}$ and consequently can be used as conditionally clean ancillae in the next rounds. Similarly, in the following round, additional Toffoli gates are implemented targeting these conditionally clean ancillae, and applying the X gate on the control qubits will again make them conditionally clean for the next round, and the process continues until we exhaust all the control qubits. Once the information from all the control qubits is accumulated in a few qubits, another (smaller) MCT gate is implemented with these qubits along with the ancilla as the control qubits and the original target as its target, so that if the ancilla was not modified in the first step, the target will also not be modified.
\begin{figure}
    \centering
    \includegraphics[width=0.5\linewidth]{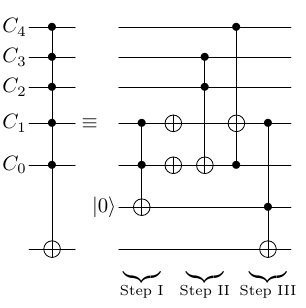}
    \caption{Understanding the conditionally clean ancillae technique from \cite{khattar24}.}
    \label{fig:ccanc}
\end{figure}

Let us now analyze Figure \ref{fig:ccanc} by distributing the possible scenarios in two mutually exclusive and exhaustive cases:

Case I: $C_0 = C_1 = C_2 = C_3 = C_4 = \ket{1}$. Then, after Step I, ancilla = $\ket{1}$. Applying X-gate on $C_0,C_1$ will make, $C_0 = C_1 = \ket{0}$, thus can be used as ancilla for Step II. Since, $C_2 = C_3 = \ket{1}$, the Toffoli gate will make, $C_0=\ket{1}$. Similarly, $C_0 = \ket{1}$, and $C_4 = \ket{1}$, imply that $C_1=\ket{1}$. Finally, $C_1=\ket{1}$, and ancilla $ = \ket{1}$, will reverse the target qubit.

Case II: At least one of $C_i=\ket{0}$, for some $i\in [0,4]$. If one of $C_0$ or $C_1$ is $\ket{0}$, then ancilla =$\ket{0}$, and the target will not be flipped in Step III. If $C_0 = C_1 = \ket{1}$, but one of $C_2, C_3, C_4$ is $\ket{0}$, then the ancilla becomes $\ket{1}$, and after applying X-gate on $C_0,C_1$ will make, $C_0 = C_1 = \ket{0}$. If one of $C_2,C_3$ is $\ket{0}$, then $C_0$ and consequently, $C_1$ also remains at $\ket{0}$ state. Similarly, if $C_4 = \ket{0}$, then also $C_1$ remains $\ket{0}$, and in both scenario, the target will not be flipped in Step III.

In Section \ref{sub:tradeoff}, we modify the design by introducing additional ancilla qubits in Step I and subsequently creating more conditionally clean ancillae to begin with, which eventually reduces the overall Toffoli depth of the complete circuit.

\section{Warm-up: A consolidated view on recent developments in Toffoli decomposition}
\label{sec:warmup}
As we have already explained, the decomposition of complex quantum gates into simpler and more practical quantum gate sets has been a topic of interest since the inception of quantum computing. More specifically, there have been substantial developments in the direction of multi-controlled Pauli (or, more precisely, MCT) decomposition in the last decades. This is a warm-up section where we provide a comprehensive overview of existing benchmarks for single and multi-controlled Toffoli gate decompositions. It emphasizes the state-of-the-art optimized results, to the best of our knowledge, in terms of T-Count, T-Depth, and ancilla requirements, as summarized in Tables \ref{tab:tof} and \ref{tab:mct}.

\subsection{Single Toffoli Decomposition}
\label{sub:tof}
To immediately dive into the technical issues, as shown in \cite[Ch. 4, Sec. 4.3]{nc}, there exists a Clifford + T decomposition of a single Toffoli gate using $7$ T gates and $9$ Clifford gates (2 H, 1 S, 6 CNOT), with a T-Depth of $6$ (Figure \ref{fig:NCTD6}). Through a simple manipulation of gate ordering, the T-Depth can be reduced to $4$ without altering the gate count. In 2013, Amy et al. \cite[Section 6]{amy13} proposed the Toffoli decomposition (Figure \ref{fig:AmyTD4}) with a T-Depth of $4$ that used one fewer Clifford gate and reduced the overall circuit depth from $12$ to $8$. In the same paper, the authors applied a meet-in-the-middle algorithm to present a Toffoli decomposition (Figure \ref{fig:AmyTD3}) with a T-Depth of $3$ and an overall depth of $9$.
For an exact decomposition of a single Toffoli gate (without measurement-based feedback), this is the lowest T-Depth that one can achieve without using any ancilla qubit.

Using one ancilla qubit, the authors also proposed a Toffoli decomposition circuit (Figure \ref{fig:AmyTD2}) using $7$ T gates and $12$ Clifford gates, achieving a T-Depth $2$. Later, in the same year, Selinger \cite[Section 2]{selinger13} proposed a Toffoli gate decomposition circuit (Figure \ref{fig:SelTD1}) with the lowest possible T-Depth of $1$, using $4$ additional ancilla qubits and $18$ Clifford gates.

\begin{figure}[htbp]
\centering
\begin{subfigure}{.45\textwidth}
  \centering
  \includegraphics[width=.9\linewidth]{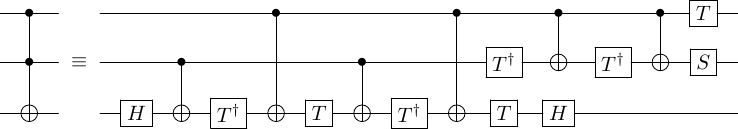}
  \caption{T-Depth: 6 \cite{nc}.}
  \label{fig:NCTD6}
\end{subfigure}\\
\begin{subfigure}{.45\textwidth}
  \centering
  \includegraphics[width=.8\linewidth]{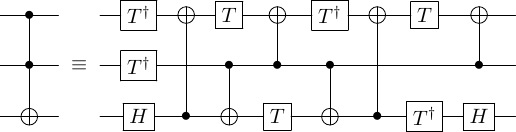}
  \caption{T-Depth: 4 \cite{amy13}.}
  \label{fig:AmyTD4}
\end{subfigure}\\
\begin{subfigure}{.45\textwidth}
  \centering
  \includegraphics[width=0.8\linewidth]{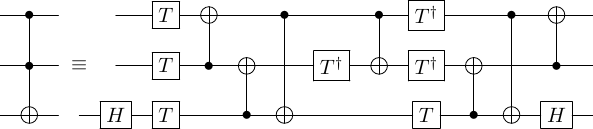}
  \caption{T-Depth: 3 \cite{amy13}.}
  \label{fig:AmyTD3}
\end{subfigure}
\caption{Single Toffoli decomposition without using any ancilla, utilizing $7$ T gates.}
\label{fig:mct}
\end{figure}

\begin{figure}[htbp]
\centering
\begin{subfigure}{.45\textwidth}
  \centering
  \includegraphics[width=0.9\linewidth]{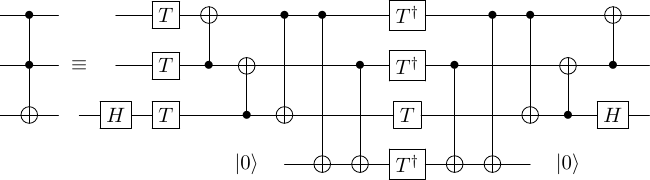}
  \caption{T-Depth: 2, Ancilla: 1 \cite{amy13}.}
  \label{fig:AmyTD2}
\end{subfigure}\\
\begin{subfigure}{.45\textwidth}
  \centering
  \includegraphics[width=0.9\linewidth]{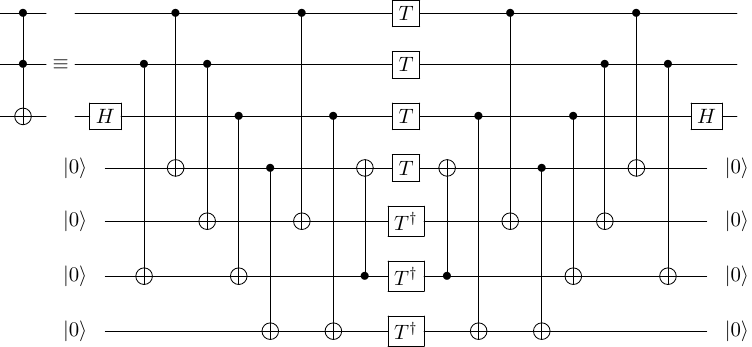}
  \caption{T-Depth: 1, Ancilla: 4 \cite{selinger13}.}
  \label{fig:SelTD1}
\end{subfigure}
\caption{Single Toffoli decomposition with ancilla, utilizing $7$ T gates.}
\label{fig:mctanc}
\end{figure}

In \cite{selinger13}, Selinger also proposed a doubly-controlled $-i$X gate, which differs from the Toffoli gate only by a controlled-$\text{S}^{\dagger}$ gate between the two control qubits. In the same year, Jones \cite{jones13} modified this circuit using a measurement-based uncomputation technique to implement the exact Toffoli gate, utilizing a single ancilla qubit with a T-Count of 4 and a T-Depth of 1. The schematic diagram of the circuit has been shown in Figure \ref{fig:jones}.
\begin{figure}
    \centering
    \includegraphics[width=0.9\linewidth]{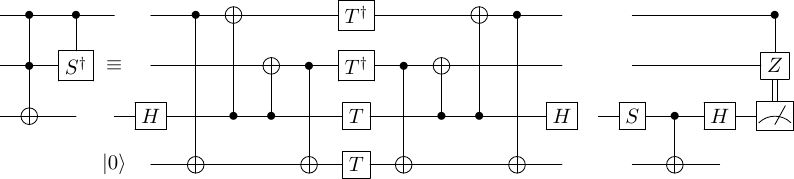}
    \caption{Single Toffoli decomposition due to Jones \cite{jones13}.}
    \label{fig:jones}
\end{figure}

Later, in 2019, Soeken integrated the designs of Selinger \cite{selinger13} and Jones \cite{jones13}, proposing a modified circuit for single Toffoli decomposition utilizing measurement-based updates \cite{jaques19}, as illustrated in Figure \ref{fig:mathias}. This circuit requires one ancilla qubit, uses four T-gates with a T-Depth of 1, and has an overall depth of 8.
\begin{figure}
    \centering
    \includegraphics[width=0.9\linewidth]{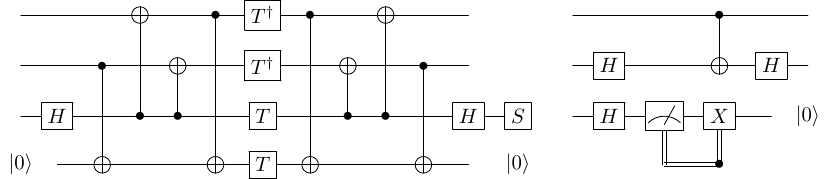}
    \caption{Single Toffoli decomposition due to Soeken \cite{jaques19}.}
    \label{fig:mathias}
\end{figure}

In 2018, Gidney introduced the concept of logical-AND \cite{gidney18} using measurement-based uncomputation, which has a T-Count of 4 and a T-Depth of 2, without using any additional ancilla. When multiple logical-AND circuits are employed within the same quantum circuit, all the initial T-gates used for preparing the state $\ket{T}$, where $\ket{T} = \text{TH}\ket{0}$, can be executed simultaneously at the beginning of the circuit, with a T-Depth of 1. Consequently, the effective T-Depth of the logical-AND decomposition reduces to 1. This is marked with (*) in the last row of the Table \ref{tab:tof}. To clarify the notation, we have
\begin{equation}
    \text{T } = \begin{pmatrix}
        1 & 0\\ 0 & e^\frac{i\pi}{4}
    \end{pmatrix},\,
    \ket{T} = \frac{1}{\sqrt{2}} \begin{pmatrix}
        1\\e^\frac{i\pi}{4}
    \end{pmatrix}.
\end{equation}
The circuit diagram for the logical-AND is presented in Figure \ref{fig:logic}.
\begin{figure}[htbp]
\centering
\begin{subfigure}{.28\textwidth}
  \centering
  \includegraphics[width=0.95\linewidth]{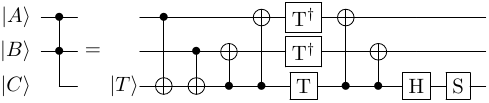}
  \label{fig:logica}
\end{subfigure}%
\begin{subfigure}{.18\textwidth}
  \centering
  \includegraphics[width=0.95\linewidth]{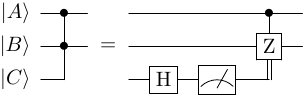}
  \label{fig:logicb}
\end{subfigure}
\caption{Single Toffoli decomposition using logical-AND due to Gidney \cite{gidney18}.}
\label{fig:logic}
\end{figure}

\begin{table*}[t]
    \centering
    \begin{tabular}{|c|c|c|c|c|c|}
    \hline
        ~Reference~ & ~Ancilla~ & ~T-Count~ & ~T-Depth~ & \#Clifford & ~Depth~\\[0.1cm]
        \hline
        \hline
        ~Amy et al. \cite{amy13}~ & {\color{red}0} & 7 & 4 & 8 & 8 \\
        \hline
        ~Amy et al. \cite{amy13}~ & {\color{red}0} & 7 & 3 & 9 & 9 \\
        \hline
        ~Amy et al. \cite{amy13}~ & 1 & 7 & 2 & 12 & 11 \\
        \hline
        ~Selinger \cite{selinger13}~ & 4 & 7 & {\color{teal}1} & 18 & 8\\
        \hline
        ~Jaques et al. \cite{jaques19}~ & 1 & {\color{blue}4} & {\color{teal}1}  & 11 & 8 \\
        \hline
        ~Gidney \cite{gidney18}*~ & {\color{red}0} & {\color{blue}4} & $1 + 1$  & 9 & 9 \\
        \hline
    \end{tabular}
    \vspace{0.3cm}
    \caption{Decompositions of a single Toffoli gate. The table specifies the required ancilla qubits, T-Count, T-Depth, Clifford counts, and the overall depth for different decompositions. The circuit with the least T-Count is highlighted in {\color{blue}blue}, the least T-Depths in {\color{teal}teal}, and the decompositions without ancilla in {\color{red}red}.}
    \label{tab:tof}
\end{table*}

Table \ref{tab:tof} summarizes various resource requirements (T-Count, T-Depth, and Ancilla count) for the Clifford + T decomposition of a single Toffoli gate.

\subsection{Multi-Controlled Toffoli Decomposition}
\label{sub:mct}
From Section \ref{sec:pre}, it is known that the multi-controlled Pauli gates (X, Y, Z) can be transformed into one another using a constant number of Clifford gates, such as the Hadamard or the Phase gates, as described in Figure \ref{fig:mcp}.
As the primary focus of this work is to minimize the implementation cost of multi-controlled Toffoli (MCT) gates in terms of Toffoli count, Toffoli depth, and the ancilla count, the inclusion of additional Clifford gates does not affect the resource estimation. Therefore, the resource estimation for any of the aforementioned multi-controlled Pauli gates can be directly translated to others without requiring modification.

In 2021, Gidney and Jones \cite{gidney21} presented a construction (Figure \ref{fig:cccz}) of a $3$-controlled Z gate using 6 T-gates having a T-Depth of 6. Additionally, they proposed that their design can be used for the construction of an $n$-controlled Pauli gate, with a T-Count of $4n-6$, using $n-2$ logical-AND gates.
\begin{figure}[htbp]
    \centering
    \includegraphics[scale=0.53]{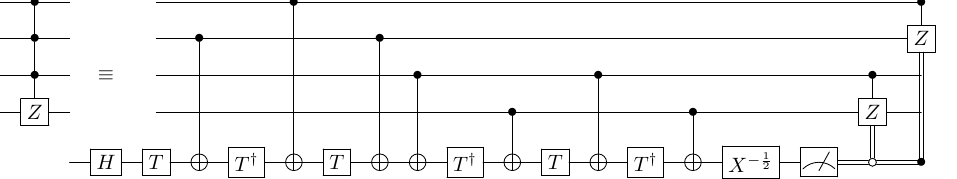}
    \caption{CCCZ circuit using 6 T gates \cite{gidney21}.}
    \label{fig:cccz}
\end{figure}
\begin{proposition}
\label{prop:cccz}
Following the CCCZ circuit of \cite{gidney21}, an $n$-MCT gate can be constructed using $n-3$ ($2$-controlled) Toffoli, and a single $3$-controlled Toffoli (CCCX) gate with $n-2$ ancilla qubits, resulting in a Toffoli depth of $\ceil{\log_2 \frac{n}{3}} +1^*$, where $1^*$ represents the depth of the CCCX gate. Replacing the Toffoli gates with logical-AND (Figure \ref{fig:logic}) yields a complete T-Depth of $\ceil{\log_2 \frac{n}{3}} +6$. Additionally, the total Clifford count of the circuit is $9n-16$.
\end{proposition}
\begin{proof}
As specified in \cite{gidney21}, out of the $n-2$ AND gates, the final AND gate and Toffoli gate can be merged to implement a CCCX gate, making the Toffoli count $(n-3)+1^*$, i.e., n-3 many ($2$-controlled) Toffoli and a single $3$-controlled Toffoli gate. Consider a tree data structure with the root node having 3-child nodes, and each of them forms a complete binary tree with a total of $n$ leaf nodes. Each of these binary trees has a depth $\ceil{\log_2 \frac{n}{3}}$, and the root node, along with its three child nodes, contributes a depth of 1. Consequently, the T-Depth becomes $\ceil{\log_2 \frac{n}{3}} +6$, where the CCCX gate contributes to the T-Depth of $6$. Since the Clifford count of each logical-AND is 9, and that of the CCCX gate 11, the total Clifford count becomes $9n-16$.
\end{proof}

In Oct 2024, Nakanishi et al. \cite{naka24} proposed a modified CCCZ circuit (Figure \ref{fig:naka}), reducing the T-Depth to 2 by utilizing an additional ancilla qubit while keeping a Clifford count of 14.
\begin{figure}[htbp]
    \centering
    \includegraphics[scale=0.35]{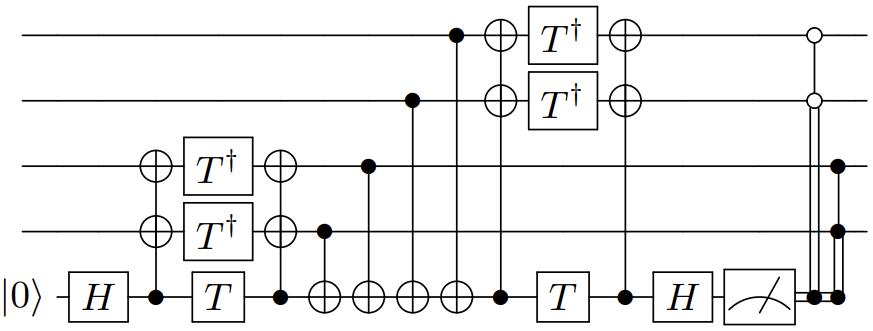}
    \caption{CCCZ circuit having T-Depth 2 \cite{naka24}.}
    \label{fig:naka}
\end{figure}
Consequently, the T-Depth of the $n$-MCT decomposition is reduced to $\ceil{\log_2 \frac{n}{3}} +2$, with the corresponding Clifford count given by $9n-15$ (see the second row of Table \ref{tab:mct}).

In Feb 2024, Nie et al. \cite{nie24} first used the notion of conditionally clean ancilla qubits derived from an existing (clean or dirty) ancilla and proposed a novel circuit decomposition (Figure \ref{fig:nie}) for the $n$-controlled Pauli gates using $\mathcal{O}(n)$ Toffoli gates. The outer layer of their MCT decomposition follows from \cite{gidney15}, and the inner layer has been parallelized to obtain an overall Toffoli depth of $\mathcal{O}(\log_2 n)$, compared to the $\mathcal{O}(\log_2 n)$ Toffoli depth in \cite{gidney15}. Additionally, by improving the design of a quantum incrementer, they developed an MCT circuit with a Toffoli count of $\mathcal{O}(n)$, and a Toffoli depth of $\mathcal{O}(\log^2 n)$ without requiring any additional ancilla qubits. In this design, although the MCT decomposition does not use any additional ancilla, implementing the quantum incrementer requires one ancilla. Since our primary focus here is to reduce Toffoli depth while increasing the clean ancilla count, we ignore the zero-ancilla implementation here.
\begin{figure}
    \centering
    \includegraphics[width=0.85\linewidth]{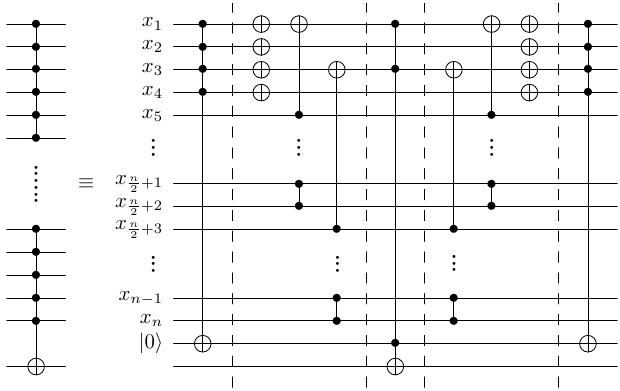}
    \caption{MCT decomposition using conditionally clean ancillae due to \cite{nie24}.}
    \label{fig:nie}
\end{figure}
\begin{proposition}
The MCT circuit decomposition proposed by \cite{nie24} using a single clean ancilla has a minimum Toffoli count of $4n+4$ and an exact Toffoli depth of $20\log_2 n$.
\end{proposition}
\begin{proof}
Step I and Step V implement 4-MCT gates following \cite{gidney15}, each requiring $4(4-2)=8$ Toffoli gates, all applied sequentially, resulting in a Toffoli depth of 16. Similarly, the 3-MCT implemented in Step III has a Toffoli count of 4 and a Toffoli depth of 4. Consequently, Steps I, III, and V together require a total of 20 Toffoli gates, with an overall Toffoli depth of 20, which is referred to in the paper as a Toffoli depth of $\mathcal{O}(1)$.

Additionally, in Step II, an $(n-4)$-MCT is decomposed into two $\left( \frac{n}{2}-2\right)$-MCT, each requiring a minimum of $2\left( \frac{n}{2}-2\right) = n-4$ Toffoli gates. Thus, the Toffoli count for Step II is $2n-8$. Similarly, Step IV also has a Toffoli count of $2n-8$. Therefore, the total Toffoli count for the entire process is given by $2(2n-8)+20=4n+4$. Moreover, from \cite{nie24}, we have $D(n)=D(n/2)+\mathcal{O}(n)$, and we estimated $\mathcal{O}(1)=20$, therefore, the exact Toffoli depth of the MCT circuit is $20\log_2 n$.
\end{proof}

In July 2024, Khattar and Gidney \cite{khattar24} further optimized the MCT circuit implementations by leveraging the conditionally clean ancilla qubits to reduce the Toffoli depth while restricting the ancilla count to $1$ or $2$. They proposed an $n$-MCT circuit utilizing a single clean ancilla (Figure \ref{fig:anc1}), achieving a Toffoli count of $2n-3$ and a T-Count of $8n-12$. Since none of the Toffoli gates are applied simultaneously, the resulting Toffoli depth is $2n-3$, while the T-Depth is $2n-3$, following the T-Depth $1$ Toffoli implementation by \cite{jaques19}, requiring one more reusable ancilla. 

Furthermore, when the availability of clean ancilla qubit increases to $2$ (Figure \ref{fig:anc2}), the Toffoli depth reduces to $\mathcal{O}(\log_2 n)$ while maintaining the Toffoli count constant. Figure \ref{fig:cca2} illustrates the circuit diagrams of $10$-MCT using $1$ and $2$ clean ancillae, due to \cite{khattar24}. Additionally, if the clean ancillae is replaced with the dirty ancillae, the Toffoli count increases to $16n-32$, and the Toffoli depth doubles under both scenarios. As we focus here on the conditionally clean ancillae derived solely from clean ancilla qubits, we do not delve into an exact analysis of the Toffoli depth for the dirty ancilla circuit implementation.

\begin{figure}[htbp]
\centering
\begin{subfigure}{.48\textwidth}
  \centering
  \includegraphics[width=0.9\linewidth]{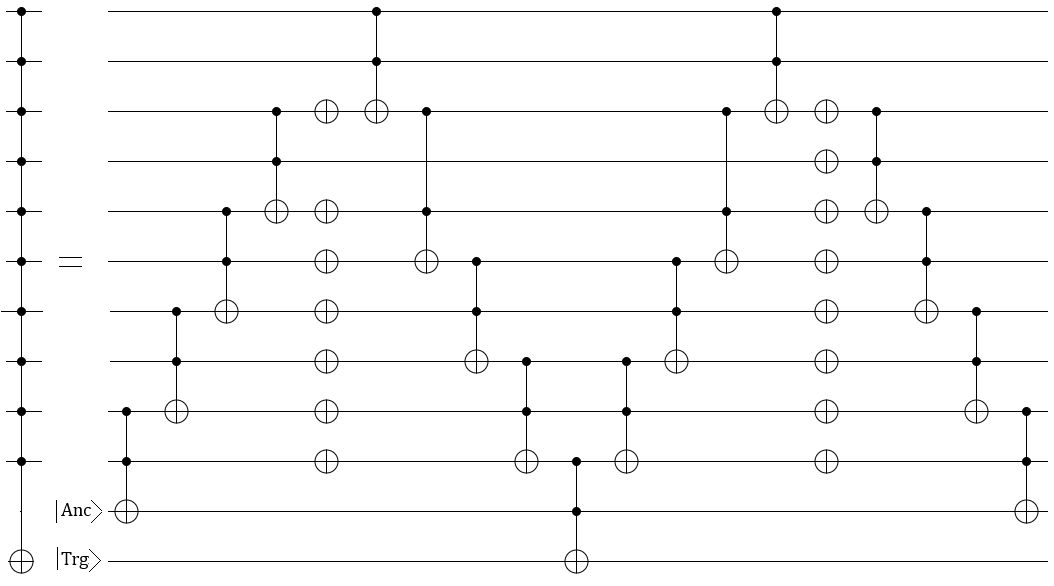}
  \caption{Ancilla: 1, Toffoli depth: 17.}
  \label{fig:anc1}
\end{subfigure}\\
\begin{subfigure}{.48\textwidth}
  \centering
  \includegraphics[width=0.9\linewidth]{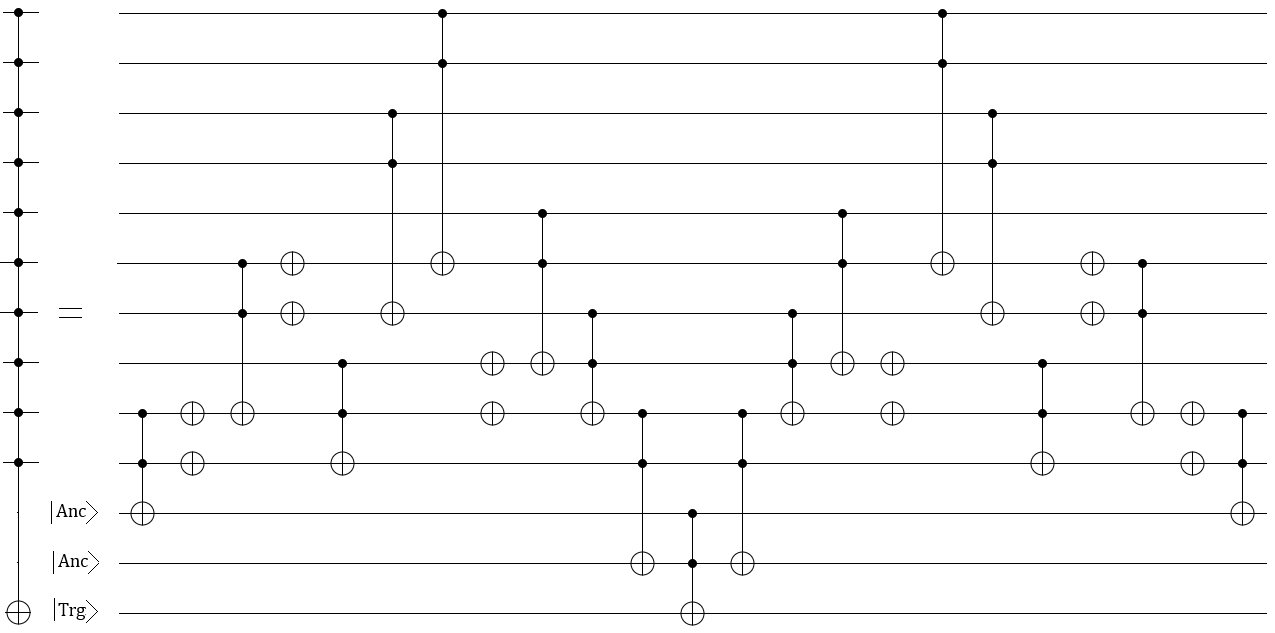}
  \caption{Ancilla: 2, Toffoli depth: 13.}
  \label{fig:anc2}
\end{subfigure}
\caption{$10$-MCT circuit decompositions, each of them using 17 Toffoli gates \cite{khattar24}.}
\label{fig:cca2}
\end{figure}
In this section, we provided the exact enumerations of Toffoli count and Toffoli depth of MCT circuit decomposition across various state-of-the-art results. For the exact enumeration of the Toffoli depth (and T-Depth), following \cite{khattar24}, we propose a novel approach of viewing the MCT decomposition with conditionally clean ancillae, which is a direct extension of \cite{khattar24}, thereby presented separately in Proposition \ref{prop:khattar} of Section \ref{sub:nvp}. Later, the approach has also been used to demonstrate the trade-off between Toffoli depth and the clean ancilla count.

The state-of-the-art optimized results corresponding to $n$-controlled Toffoli decompositions have been summarized in Table \ref{tab:mct}.
\begin{table*}[t]
\adjustbox{max width=\textwidth}{
    \centering
    \begin{tabular}{|c|c|c|c|c|c|}
    \hline
        Reference & ~Ancilla~ & Toffoli count & Toffoli depth & ~T-Count~ & ~T-Depth~ \\[0.1cm]
        \hline
        \hline
        Gidney et al. \cite{gidney21}~ & $n-2$ & $(n-3)+1^*$ & $\ceil{\log_2 \frac{n}{3}} +1^*$ & {\color{blue}$4n-6$} & $\ceil{\log_2 \frac{n}{3}}+6$\\
        \hline
        Nakanishi et al. \cite{naka24}~ & $n-1$ & $(n-3)+1^*$ & $\ceil{\log_2 \frac{n}{3}} +1^*$ & {\color{blue}$4n-6$} & {\color{teal}$\ceil{\log_2 \frac{n}{3}}+2$}\\
        \hline
        Nie et al. \cite{nie24}~ & {\color{red}1} & $\geq 4n+4$ & $20\log_2 n$ & $16n+16$ & $20\log_2 n$\\
        \hline
        Khattar et al. \cite{khattar24} & {\color{red}1} & $2n-3$  & $2n-3$ & $8n-12$ & $2n-3$\\
        \hline
        Khattar et al. \cite{khattar24}~ & 2 & $2n-3$  & $\approx 4\log_2 n$ & $8n-12$ & $\approx 4\log_2 n$\\
        \hline
        Ours & $m_1(\ll n)+2$ & $2n-m_1-3$ & $\begin{cases}
    2\floor{\log_2 m_1} +4k, \text{if }n\in \left[\left( m_1 + 2^{\floor{\log_2 m_1}-1}\right) 2^{k-1}\right.\\
    \hspace{4cm}\left.+k-1, m_12^k + k -2\right],\\
    2\floor{\log_2 m_1} +4k +2, \text{if }n\in \left[ m_12^k + k -1,\right.\\
    \hspace{2.5cm}\left.\left(m_1 + 2^{\floor{\log_2 m_1}-1} \right)2^{k} +k -1\right], k\geq 2.
\end{cases}$ & $8n-4m_1-12$ & Same as Toffoli depth\\
        \hline
    \end{tabular}
}
    \vspace{0.3cm}
    \caption{Decompositions of $n$-controlled Toffoli gates. The table specifies the required ancilla qubits, Toffoli count, Toffoli depth, T-Count, and T-Depth. The circuit with the least T-Count is highlighted in {\color{blue}blue}, the least T-Depths in {\color{teal}teal}, and the decompositions with the least ancilla in {\color{red}red}.}
    \label{tab:mct}
\end{table*}

\section{Further Reduction of Toffoli Depth}
\label{sec:cont1}
In this section, we adopt the MCT decomposition technique using the conditionally clean ancillae, as introduced in \cite{nie24,khattar24}, and propose a trade-off between the Toffoli depth with the clean ancilla, showing improvement over the recent work of Khattar and Gidney \cite{khattar24}. The authors proposed the MCT circuit decompositions techniques, achieving a Toffoli depth of $\mathcal{O}(n)$ with a single clean ancilla and $\mathcal{O}(\log_2 n)$ with two clean ancillae, where in both cases, the Toffoli count becomes $2n-3$. Here, we revisit the MCT decompositions by \cite{khattar24} and provide an exact enumeration of the Toffoli depth (consequently T-Depth) that was not explicitly presented in~\cite{khattar24}.
\subsection{Exact Enumeration of Toffoli Depth proposed by Khattar and Gidney (2024) in a Different Lens}
\label{sub:nvp}

In this subsection, we enumerate the exact Toffoli depth of an $n$-MCT circuit following \cite{khattar24}. The $n$-MCT decomposition utilizing a single ancilla requires $2n-3$ Toffoli gates, all applied sequentially, resulting in a Toffoli depth of $2n-3$. Subsequently, by employing the measurement-based uncomputation technique for Toffoli decomposition, as illustrated in Figure \ref{fig:mathias}, the T-Depth is again the same, i.e., $2n-3$.

The $n$-MCT decomposition circuit \cite{khattar24} using two ancilla qubits consists of five steps. Step I and Step V combined have a Toffoli depth of 1 (using logical-AND). In Step II, information from the remaining $n-2$ control qubits is accumulated into $k(\ll n)$ qubits using the conditionally cleaned ancillae technique. Next, in Step III, a $(k+1)$-MCT is implemented using the unused ancilla, following the single-ancilla technique, having a Toffoli depth of $2(k+1)-3 = 2k-1$. In Step IV, the control qubits are returned to their original state through uncomputation, thereby having the same Toffoli count and depth as in Step II.

For an exact enumeration of the Toffoli depth, we introduce a new way of viewing Step I and Step II of the MCT circuit decomposition, as follows.
\begin{align*}
    2 \rightarrow & 1 && && && && && && && && && &&\\
    & _{+1}^{2} \rightarrow && _{+1}^{1} \rightarrow && 1 && && && && && &&\\
    & && _{+1}^{4} \rightarrow && _{+1}^{2} \rightarrow && _{+1}^{1} \rightarrow && 1 && && && &&\\
    & && && _{+1}^{8} \rightarrow && _{+1}^{4} \rightarrow && _{+1}^{2} \rightarrow && _{+1}^{1} \rightarrow && 1 && &&\\
    & && && && _{+1}^{12} \rightarrow && _{+1}^{6} \rightarrow && _{+1}^{3} \rightarrow && 2 \rightarrow && 1 &&
\end{align*}
where the first row corresponds to Step I, and the remaining rows correspond to Step II. Each layer of horizontal arrows represents a Toffoli depth of 1, and
a row, $a \rightarrow \ceil{a/2} \rightarrow \ldots \rightarrow 1$ implies that the information from $a$ many qubits has been accumulated into a single qubit, with a Toffoli depth of $\ceil{\log_2 a}$. Consequently, the first element of each row sums to $n$. Additionally, the number of layers ($k$) determines the size of the MCT gate required for Step III.

In the above example, Step I and Step II of a $32$-MCT decomposition have been shown using the conditionally cleaned ancillae technique. The Step I has a Toffoli depth of $1$, and Step II has a Toffoli depth of $7$. Further, in Step III, a $5$-MCT needs to be applied to hit the target qubit, having a Toffoli depth $2(5)-3 = 7$. However, a part of the smaller MCT is applied simultaneously with Step II and Step IV, thereby having an effective Toffoli depth of $3$ or $5$. Clearly, the uncomputation (Step IV) requires an additional Toffoli depth of $8$, i.e., the Toffoli depth of the complete circuit is $19$.

Using this new approach, the $5$-MCT circuit using a conditionally clean ancillae technique, presented in Figure \ref{fig:ccanc}, can be seen as follows.
\begin{align*}
    2 \rightarrow & 1 && && && &&\\
    & _{+1}^{2} \rightarrow && _{+1}^{1} \rightarrow && 1,\\
\end{align*}
where the Toffoli depth for Step I and Step II combined is 3, which can also be verified from the actual circuit from Figure \ref{fig:ccanc}. Moreover, a $2$-controlled Toffoli needs to be implemented in Step III to modify the target.

\begin{remark}
    The motivation for viewing the MCT decomposition using this new representation because it allows more efficient and accurate estimation of the Toffoli depth, compared to the standard circuit representation. Additionally, our approach provides insight into the number of simultaneous operations and the size of the MCT gate required for implementation in Step III. In the subsequent section, we present our MCT circuit design with additional ancilla qubits using the same representation.

    Further, note that introducing additional rows in Step II, without increasing the overall Toffoli depth will increase the size of the MCT gate required in Step III. Consequently, the overall Toffoli depth of the circuit remains unchanged.
\end{remark}

In the direction of enumerating the exact Toffoli depth of an $n$-MCT decomposition due to \cite{khattar24}, using 2 clean ancillae, we have the following lemma.
\begin{lemma}
The exact Toffoli depth, $\delta$ for Step II or Step IV of an $n$-MCT circuit decomposition, as proposed by \cite{khattar24}, using two clean ancilla qubits, varies with $n$ as follows.
$$\delta=
\begin{cases}
    2k-3, \, \text{ for } n\in\left[3\cdot 2^{k-2}+k-1,\, 2^k + k -2\right],\\
    2k-2, \, \text{ for } n\in \left[2^k + k -1,\, 3\cdot 2^{k-1} +k -1\right]
\end{cases}
$$
where $k\in\mathbb{N}$, with $k\geq 2$.
\end{lemma}
\begin{lemma}
    The number of control qubits, $\sigma$, in the MCT gate implemented in Step III of an $n$-MCT circuit decomposition, as proposed by \cite{khattar24}, using two clean ancilla qubits, varies with $n$ as follows.
    $$\sigma = \left\{\begin{array}{lr}
        k & \,\,\text{for } n\in \left[ 3\cdot 2^{k-2}+k-1,\, 2^k + k -1\right],\\
        k+1 & \,\,\text{for } n\in \left( 2^k + k -1,\, 3\cdot 2^{k-1} +k -1\right].
        \end{array}\right.$$
    where $k\in\mathbb{N}$, with $k\geq 2$.
    The corresponding Toffoli depths are $2k-3$, and $2(k+1)-3 = 2k-1$, respectively. However, a significant part of the smaller MCT is applied simultaneously with Step II and Step IV, making the effective Toffoli depth of Step III, either $3$ or $5$.
\end{lemma}

From the above lemma, we can now estimate the exact Toffoli depth of the $n$-MCT circuit decomposition using 2 ancilla qubits, as proposed in \cite{khattar24}. 
\begin{proposition}
\label{prop:khattar}
    Following \cite{khattar24}, the exact Toffoli depth of an $n$-MCT decomposition using $2$ clean ancillae is lower bounded by $3\log_2 n$.
\end{proposition}
\begin{proof}
    From the above lemma, the exact Toffoli depth, $\delta$, of the complete $n$-MCT circuit decomposition is given by
    $$
    \delta = 
    \left\{\begin{array}{lr}
        1 + 2(2k-3) + 3 = 4k-2 &\\
        \hspace{1.8cm} \text{for } n\in \left[3\cdot 2^{k-2}+k-1, \, 2^k + k -1\right],\\
        1 + 2(2k-2) + 3 = 4k &\\
        \hspace{1.8cm} \text{for } n\in \left(2^k + k -1,\, 3\cdot 2^{k-1} +k -1\right]
        \end{array}\right.
    $$
    where $k\in\mathbb{N}$, with $k\geq 2$. Since here, we are interested in showing the lower bound, we consider $n$ assumes the highest value in the range and still show that $3\log_2 n$ is strictly less than the corresponding depths, as follows.
        $$3\log_2 (2^k + k -1) < 3\log_2 (2^{k+1}) = 3(k+1),$$
    which is less or equal to the depth $4k-2$, for $k\geq 5$.
    $$3\log_2 (3\cdot 2^{k-1} +k -1) < 3\log_2 (4\cdot 2^{k-1}) = 3(k+1),$$
    which is less or equal to the depth $4k$ for $k\geq 3$.
\end{proof}

Although the above proposition establishes that the exact Toffoli depth using two clean ancillae, as per \cite{khattar24}, is lower bounded by $3\log_2 n$, this bound is not tight and is closer to $4\log_2 n$. In the following section, we analyze the trade-off between Toffoli depth and clean ancilla, demonstrating that with additional ancilla qubit, the exact Toffoli depth can be reduced to $2\log_2 n$.

\subsection{Exact Trade-Off between Toffoli Depth and Clean Ancillae}
\label{sub:tradeoff}
In this subsection, we explore the ancilla-Toffoli depth trade-off using the concept of conditionally clean ancillae and demonstrate improvements in the Toffoli depth compared to \cite{khattar24} by introducing additional ancilla qubits into the circuit while keeping the Toffoli count constant. Furthermore, for an $n$-controlled Toffoli gate with $m$ ancilla qubits, we propose an algorithm that determines the Toffoli depth (and consequently the T-Depth) for the MCT circuit decomposition and presents a graph illustrating the trade-off.

\begin{construction}
\label{cons:tradeoff}
Given $m$ ancilla qubits to implement an $n$-controlled Toffoli gate, we distribute $m$ as $m_1 + m_2$, where $m_1$ ancilla qubits are allocated for Step I, and $m_2$ ancilla qubits are reserved for Step III. For $m=3$, we set $m_1=2$ and $m_2=1$. Furthermore, for $m\geq 4$, we assume $m_2=2$ to implement the smaller MCT in Step III using the $2$-clean ancillae technique described in \cite{khattar24}.

The circuit design is as follows.
\begin{itemize}
    \item In Step I, the $m_1$ Toffoli gates are applied to collect information from $2m_1$ control qubits and store them in $m_1$ ancilla qubits. We implement them using the logical-AND circuit to avoid resource requirements for uncomputation in the final step.
    \item In Step II, information from the remaining control qubits is accumulated into $k(< n)$ qubits using the conditionally cleaned ancillae.
    \item  In Step III, a $(k + m_1)$-controlled Toffoli gate is implemented using the remaining $m_2$ many ancillae, following the $2$-ancilla technique, and the target qubit is modified.
    \item In Step IV, uncomputation is performed to return the control qubits to their original state. As earlier, the Toffoli count and Toffoli depth of Step IV are identical to those in Step II.
    \item Finally, in Step V, the first $m_1$ many Toffoli gates are uncomputed to clean up the ancillae, using logical-AND based uncomputation. This step reduces the overall Toffoli count by $m_1$ compared to the circuit presented in \cite{khattar24}.
\end{itemize}

The schematic diagram of the MCT decomposition using $m = m_1+m_2$ many ancillae is similar to the diagram we used for the exact Toffoli depth enumeration in Section \ref{sub:nvp}.

\begin{widetext}
\begin{align*}
    2m_1 \rightarrow & m_1 && && && && && && && && && &&\\
    &_{+1}^{2m_1} \rightarrow && _{+1}^{m_1} \rightarrow && \ldots \rightarrow && 1 && && && && && && &&\\
    & && _{+1}^{4m_1} \rightarrow && _{+1}^{2m_1} \rightarrow && _{+1}^{m_1} \rightarrow && \ldots \rightarrow && 1 && && && && &&\\
    & && && _{+1}^{8m_1} \rightarrow && _{+1}^{4m_1} \rightarrow && _{+1}^{2m_1} \rightarrow && \ldots \rightarrow && 1 && && && &&\\
    & && && && \vdots \rightarrow && \vdots \rightarrow && \vdots \rightarrow && \vdots \rightarrow && \ldots \rightarrow && 1 && &&
\end{align*}
\end{widetext}

\qed
\end{construction}

\begin{lemma}
For the proof of correctness, the circuit decomposition presented in Construction \ref{cons:tradeoff} is an $n$-MCT.
\end{lemma}
\begin{proof}
    In the above MCT decomposition, if any of the first $2m_1$ qubits is in the $\ket{0}$ state, one of the $m_1$ ancilla qubits will also remain in $\ket{0}$, preventing the target flip in Step III. Similarly, if any of the $n-2m_1$ qubits is $\ket{0}$, the corresponding row's end qubit remains in $\ket{0}$, ensuring the target is not flipped in Step III. This proves that the designed circuit correctly implements an $n$-MCT gate.
\end{proof}

\begin{example}
\label{ex:1}
    For $m_1=3$, the $32$-controlled Toffoli can be viewed as:
\begin{align*}
    6 \rightarrow & 3 && && && && &&\\
    &_{+1}^{6} \rightarrow && _{+1}^{3} \rightarrow && 2 \rightarrow && 1 && &&\\
    & && 12 \rightarrow && 6 \rightarrow && 3 \rightarrow && _{+1}^{1} \rightarrow && 1\\
    & && && _{+1}^{6} \rightarrow && _{+1}^{3} \rightarrow && 2 \rightarrow && 1
\end{align*}
where the Toffoli depth for Step I is $1$, and the Toffoli depth for Step II (and Step IV) are $5$ each. Additionally, in Step III, we need to implement a $6$-controlled Toffoli using $m_2=2$ ancillae, which requires an additional Toffoli depth of $3$. Consequently, the Clifford + Toffoli decomposition of the $32$-controlled Toffoli gate, utilizing $(3+2) = 5$ ancilla qubits, achieves a total Toffoli depth of $1 + 2(5) + 3 = 14$. In contrast, the $32$-controlled Toffoli decomposition with $2$ ancilla qubits, as presented in \cite{khattar24}, results in a Toffoli depth of $19$.

Similarly, for $m_1=4$, the $32$-controlled Toffoli can be viewed as:
\begin{align*}
    8 \rightarrow & 4 && && && && &&\\
    &_{+1}^{8} \rightarrow && _{+1}^{4} \rightarrow && _{+1}^{2} \rightarrow && _{+1}^{1} \rightarrow && 1 &&\\
    & && _{+1}^{14} \rightarrow && _{+1}^{7} \rightarrow && 4 \rightarrow && 2 \rightarrow && 1
\end{align*}
where the Toffoli depth remains $14$. Thus, it can be observed that beyond a certain point, further increases in the number of ancilla qubits do not necessarily lead to a reduction in the Toffoli depth when using conditionally clean ancillae.
\end{example}

The Toffoli depth of the MCT decomposition using $m=m_1+m_2$ many ancillae can be estimated from the following lemma.

\begin{lemma}
\label{lem:step2}
The exact Toffoli depth of Step II or Step IV of an $n$-MCT circuit decomposition, implemented above, using $m=m_1+m_2$ clean ancillae following the conditionally clean ancillae technique, varies with $n$ as follows.
$$\delta=
\begin{cases}
    \floor{\log_2 m_1} + 2k-3,\\
    \indent \text{for }n\in \left[\left( m_1 + 2^{\floor{\log_2 m_1}-1}\right) 2^{k-1}+k-1,\right.\\
    \hspace{5.2cm}\left.m_12^k + k -2\right]\\
    \floor{\log_2 m_1} + 2k-2,\\
    \indent\text{for }n\in \left[ m_12^k + k -1,\right.\\
    \hspace{2.6cm}\left. \left(m_1 + 2^{\floor{\log_2 m_1}-1} \right)2^{k} +k -1\right]
\end{cases}
$$
where $k\in\mathbb{N}$, with $k\geq 2$.
\end{lemma}
\begin{lemma}
\label{lem:step3}
    The number of control qubits, $\sigma$, in Step III of an $n$-MCT circuit decomposition, implemented above using $m=m_1+m_2$ ancilla qubits, varies with $n$ as follows.
    $$\sigma = \left\{\begin{array}{l}
        k+m_1-1,\\
        \indent \text{ for } n\in \left[\left( m_1 + 2^{\floor{\log_2 m_1}-1}\right) 2^{k-1}+k-1,\right.\\
        \hspace{5.2cm}\left.m_12^k + k -1\right]\\
        k+m_1,\\
        \indent\text{ for }n\in \left(m_12^k + k -1,\right.\\
        \hspace{2.6cm}\left.\left(m_1 + 2^{\floor{\log_2 m_1}-1} \right)2^{k} +k -1\right]
        \end{array}\right.$$
    where $k\in\mathbb{N}$, with $k\geq 2$.
    The corresponding Toffoli depths are $4\log_2 (k+m_1-1)$, and $4\log_2(k+m_1)$, respectively. However, a significant part of the smaller MCT is applied simultaneously with Step II and Step IV, making the effective Toffoli depth of Step III, either $3$ or $5$.
\end{lemma}

From the above lemma, we can now estimate the exact Toffoli depth of the $n$-MCT circuit decomposition using $m=m_1+m_2$ ancilla qubits, as presented above. 
\begin{theorem}
\label{prop:tradeoff}
    The exact Toffoli depth, $\delta$, of the complete $n$-MCT circuit decomposition using $m=m_1+m_2$ clean ancillae is given by
    $$\delta=
\begin{cases}
    2\floor{\log_2 m_1} +4k,\\
    \indent \text{for }n\in \left[\left( m_1 + 2^{\floor{\log_2 m_1}-1}\right) 2^{k-1}+k-1,\right.\\
    \hspace{5.2cm}\left.m_12^k + k -2\right]\\
    2\floor{\log_2 m_1} +4k +2,\\
    \indent \text{for }n\in \left[ m_12^k + k -1,\right.\\
    \hspace{2.6cm}\left.\left(m_1 + 2^{\floor{\log_2 m_1}-1} \right)2^{k} +k -1\right]
\end{cases}
$$
where $k\in\mathbb{N}$, with $k\geq 2$.
\end{theorem}
\begin{proof}
    From Lemma \ref{lem:step2}, and \ref{lem:step3}, the exact Toffoli depth of an $n$-MCT, from Step I-IV, for $n\in \left[\left( m_1 + 2^{\floor{\log_2 m_1}-1}\right) 2^{k-1}+k-1, m_12^k + k -2\right]$ is
    $$1 +2\left(\floor{\log_2 m_1} + 2k-3\right)+5 = 2\floor{\log_2 m_1} +4k.$$

    For $n\in \left[ m_12^k + k -1,\left(m_1 + 2^{\floor{\log_2 m_1}-1} \right)2^{k} +k -1\right]$, the exact Toffoli depth of the $n$-MCT becomes 
    $$1 +2\left(\floor{\log_2 m_1} + 2k-2\right)+5 = 2\floor{\log_2 m_1} +4k +2.$$
\end{proof}



It is now understood that the depth of $n$-MCT can be achieved in $\mathcal{O}(\log_2 n)$, which can be explicitly written $c\log_2 n$, where $c$ needs to be properly estimated, for exact trade-offs which is the main motivation in this paper. In Proposition \ref{prop:khattar}, we have shown that in the MCT decomposition using two clean ancillae with the technique from \cite{khattar24}, $c$ is strictly greater than 3; in fact, it is around 4. However, using the ancilla - Toffoli depth trade-off, $c$ can be further reduced to a certain extent. In the following subsection, we demonstrate that with the conditionally clean ancillae technique, $c$ always remains strictly greater than 1, regardless of the number of ancilla qubits used.

\subsection{Proving the Lower Bound on Toffoli Depth using Conditionally Clean Ancillae}
\label{sub:lowerbound}
In this subsection, we show that the Toffoli depth of an $n$-MCT using the conditionally clean ancillae technique can never be reduced to $\ceil{\log_2 n}$.
\begin{theorem}
This is in reference to Construction \ref{cons:tradeoff} for $n$-MCT.
\begin{enumerate}
    \item Assuming that the control states do not need to be restored to their original values, using the conditionally clean ancillae technique, the exact Toffoli depth must be strictly greater than $\ceil{\log_2 n}$, irrespective of the number of available ancilla.
    \item When the control qubits are required to be returned to their original state upon completion, the Toffoli depth becomes strictly greater than $2\ceil{\log_2 n}$.
\end{enumerate}
\end{theorem}
\begin{proof}
From Proposition \ref{prop:tradeoff}, it is evident that the Toffoli depth remains constant over a range of values for $n$. Since here, we are interested in showing the lower bound, we consider $n$ assumes the highest value in the range and still show that $\log_2 n$ is strictly less than the corresponding depth, as follows. 

Let us first proof the item 1. As the control states are not required to be returned to their original state, we consider the Toffoli depth due to Step I, Step II, and half of Step III only, which is 
$$\delta'=
\begin{cases}
    \floor{\log_2 m_1} + 2k +1, \text{if }n\in \left[\left( m_1 + 2^{\floor{\log_2 m_1}-1}\right) 2^{k-1}\right.\\
    \hspace{4cm}\left.+k-1, m_12^k + k -2\right]\\
    \floor{\log_2 m_1} + 2k + 2, \text{if }n\in \left[ m_12^k + k -1,\right.\\
    \hspace{3cm}\left.\left(m_1 + 2^{\floor{\log_2 m_1}-1} \right)2^{k} +k -1\right]
\end{cases}
$$
Thus, we need to show the following two cases:
\begin{itemize}
    \item Case I: $\ceil{\log_2\left( m_12^k +k -2\right)} \leq \floor{\log_2m_1} + 2k +1$,
    \item Case II: $\ceil{\log_2\left(\left(m_1 +2^{\floor{\log_2 m_1}-1} \right)2^k +k -1\right)}\leq \floor{\log_2m_1} + 2k+2$.
\end{itemize}
Case I: Showing $\ceil{\log_2\left( m_12^k +k -2\right)} \leq \floor{\log_2m_1} + 2k +1$ is equivalent to showing $m_12^k +k -2 \leq 2^{ \floor{\log_2m_1} + 2k}$. We prove this by induction on $k\geq 2$, for $k\in\mathbb{N}$.\\
\\
\textbf{Base case}: For $k=2$, RHS = $2^{\floor{\log_2 m_1}+4} > 2^{\log_2 m_1-1+4} = 8m_1$, which is strictly greater than $4m_1$, LHS for $k=2$.\\
\\
\textbf{Induction hypothesis}: Suppose the result holds for $k=k_1\in\mathbb{N}$, i.e.,
$$m_12^{k_1} +k_1 -2 \leq 2^{\floor{\log_2m_1} + 2k_1}.$$

\noindent \textbf{Inductive step}: We need to show that the result holds for $k=k_1+1$, i.e.,
$$m_12^{k_1+1} + k_1 -1 \leq 2^{\floor{\log_2m_1} + 2 (k_1+1)}.$$
LHS, $m_1 2^{k_1+1} +k_1 - 1 = 2\left(m_12^{k_1} +k_1 -2\right) - k_1+3 \leq 2 \left(2^{\floor{\log_2m_1} + 2k_1} \right) - k_1+3,$ which is strictly less than $4 \left(2^{\floor{\log_2m_1} + 2k_1} \right) = 2^{\floor{\log_2m_1} + 2(k_1+1)}$, RHS.

Therefore, Case I holds for all $k \geq 2$ for some $k \in \mathbb{N}$.\\

\noindent Case II: Showing $\ceil{\log_2\left(\left(m_1 +2^{\floor{\log_2 m_1}-1} \right)2^k +k -1\right)}\leq \floor{\log_2m_1} + 2k+2$ is equivalent to showing
$$\left(m_1 +2^{\floor{\log_2 m_1}-1} \right)2^k +k -1 \leq 2^{\floor{\log_2m_1} + 2k +1}.$$
We again prove this by induction on $k\geq 2$, for $k\in\mathbb{N}$.\\

\noindent \textbf{Base case}: For $k=2$, LHS = $4\left(m_1 + 2^{\floor{\log_2 m_1}-1} \right) + 1 = 4m_1 + 1 + 2^{\floor{\log_2 m_1}+1}$. Similarly, for $k=2$, RHS = $2^{\floor{\log_2 m_1}+4 +1} = 16\cdot 2^{\floor{\log_2 m_1}+1}$, which can be distribute into $15\cdot 2^{\floor{\log_2 m_1}+1} + \cdot 2^{\floor{\log_2 m_1}+1}$.

Now, $15\cdot 2^{\floor{\log_2 m_1}+1} > 15 \cdot 2^{\log_2 m_1-1+1} = 15m_1> 4m_1 + 1$. Thus, the base case holds.\\

\noindent \textbf{Induction hypothesis}: Suppose the result holds for $k=k_1 \in\mathbb{N}$, i.e.,
$$\left(m_1 +2^{\floor{\log_2 m_1}-1} \right)2^{k_1} +k_1 -1 \leq 2^{\floor{\log_2m_1} + 2k_1+1}.$$

\noindent \textbf{Inductive step}: We need to show that the result holds for $k=k_1+1$, i.e.,
$$\left(m_1 +2^{\floor{\log_2 m_1}-1} \right)2^{k_1+1} + k_1\leq 2^{\floor{\log_2m_1} + 2k_1+3}.$$
\begin{align*}
\text{LHS = } & \left(m_1 +2^{\floor{\log_2 m_1}-1} \right)2^{k_1+1} +k_1\\
= & 2\left[\left(m_1 +2^{\floor{\log_2 m_1}-1} \right)2^{k_1} +k_1 -1\right] - k_1+2\\
\leq & 2 \left(2^{\floor{\log_2m_1} + 2k_1 +1} \right) - k_1+2 \\
< & 4 \left(2^{\floor{\log_2m_1} + 2k_1 +1} \right) = 2^{\floor{\log_2m_1} + 2k_1+3} = \text{ RHS}.
\end{align*}\

Therefore, Case II holds for all $k\geq 2$ for some $k\in\mathbb{N}$. In conclusion, using the conditionally clean ancillae technique, the exact Toffoli depth of an $n$-controlled Toffoli decomposition can not be reduced to $\ceil{\log_2 n}$, even when the control qubits are not restored to their original states.

Item 2, where the control qubits are restored to their original state, can be proved inductively in a similar manner. The factor of 2 needs to be multiplied in this case because the full depth, $\delta$, from Proposition \ref{prop:khattar}, is related to $\delta'$ as $\delta = 2\delta' -2$, i.e., almost twice the depth considered in the first part.
\end{proof}

The above theorem shows that using the conditionally clean ancillae, the constant factor $c$ can never be reduced to $1$. In the next section, we demonstrate that the exact Toffoli depth in the Clifford + Toffoli decomposition of an $n$-MCT is lower bounded by $\ceil{\log_2 n}$, regardless of the technique or the number of ancilla qubits used. This bound is achievable through complete binary tree decomposition of $n$-controlled Toffoli gates.

\section{Tight $\ceil{\log_2 n}$ Lower Bound on Toffoli Depth}
\label{sec:cont2}
In this section, we show in a more general framework that the exact Toffoli depth in the Clifford + Toffoli decomposition of an $n$-controlled Toffoli gate is lower bounded by $\ceil{\log_2 n}$, which is exactly $\log_2 n$ when $n = 2^k$ for some $k\in\mathbb{N}$. Additionally, the exact T-Depth also becomes $\ceil{\log_2 n}$, utilizing $2n - 2$ ancilla qubits and a T-Count of $4(n-1)$, provided by \cite{jones13}. Alternatively, following Gidney’s logical-AND circuit \cite{gidney18}, the ancilla count can be reduced to $n-2$, maintaining the same T-Count, while the T-Depth increases to $\ceil{\log_2 n} + 1$.

\begin{theorem}
\label{th:logn}
The exact Toffoli depth in the Clifford + Toffoli decomposition of an $n$-MCT is lower bounded by $\ceil{\log_2 n}$, given any technique or any number of ancilla qubits used.
\end{theorem}
\begin{proof}
Each $2$-controlled Toffoli gate encodes the information of two control qubits into one target qubit. This process can be visualized as a binary tree, where the leaf nodes represent the $n$ control qubits, and their information is successively accumulated in internal parent nodes.

To implement an $n$-MCT, we begin with $n$ leaf nodes, where each pair of control qubits is combined using a Toffoli gate, reducing the number of active qubits in each round. At each stage, either two parent nodes or one parent node, along with a leftover node from the previous round, are further combined into a new parent node. This process continues iteratively until the final Toffoli gate transfers the accumulated information to the root node (the target qubit).

Since each Toffoli gate reduces two qubits into one, the number of required levels in the binary tree is given by the height of a binary tree with $n$ leaf nodes. The depth of such a binary tree is at least $\ceil{\log_2 n}$. Hence, the Toffoli depth of the $n$-MCT is also at least $\ceil{\log_2 n}$, proving the claim.
\end{proof}

A complete binary tree structure, having $n$ leaf nodes, has a height $\log_2 n$, i.e., when $n$ is not a power of $2$. When $n=2^k$, for some $k\in\mathbb{N}$, then it becomes a perfect binary tree, having the depth exactly $\log_2 n = k$. In both cases, the number of internal nodes, i.e., ancilla qubit, is $n-2$. Additionally, the Toffoli count for both these cases becomes $n-1$.

For example, a $32$-MCT can be implemented using $30$ ancilla qubits, with $33$ Toffoli gates, having a minimum Toffoli depth of $5$. Similarly, all $n$-MCT decomposition, with $17\leq n \leq 32$, can be implemented with a minimum Toffoli depth of $5$. Figure \ref{fig:mct7} provides a schematic diagram of the $7$-MCT circuit decomposition using $5$ ancilla qubits, and $6$ Toffoli gates, with a Toffoli depth of $\ceil{\log_2 7}=3$. In the diagram, the first four Toffoli gates in the first two columns are implemented simultaneously, having depth 1. The next two Toffoli gates are applied sequentially, having a Toffoli depth of 1 each.
\begin{figure}
    \centering
    \includegraphics[width=0.35\linewidth]{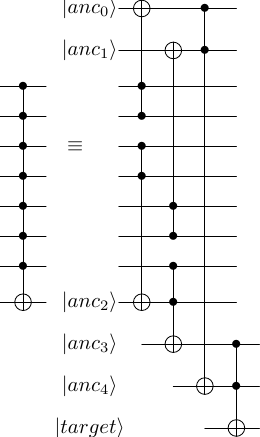}
    \caption{$7$-MCT circuit decomposition with a Toffoli depth 3.}
    \label{fig:mct7}
\end{figure}

In this context, we present the following corollaries concerning the lower bound on the T-depth of an $n$-MCT circuit implementation via Clifford+Toffoli decomposition, regardless of the number of ancilla qubits used. These results can be obtained by first constructing an $\ceil{\log_2 n}$ Toffoli depth circuit for the $n$-MCT decomposition, followed by further decomposing the Toffoli gates into Clifford+T gates as outlined in Table \ref{tab:tof}.

\begin{corollary}
Using the measurement-based Toffoli decomposition circuit proposed by \cite{jaques19}, the T-depth of the Clifford+T decomposition of an $n$-MCT gate, implemented via Clifford+Toffoli decomposition, is lower bounded by $\ceil{\log_2 n}$. Furthermore, the circuit requires $2n-2$ ancilla qubits, and the T-count is $4n-4$.
\end{corollary}
\begin{corollary}
Using the logical-AND circuit proposed by \cite{gidney18}, an $n$-MCT gate can be implemented with a Clifford+T decomposition utilizing $n-2$ ancilla qubits. The resulting circuit has a T-depth of $\ceil{\log_2 n}+1$ and a T-count of $4n-4$.
\end{corollary}

Towards, the conclusion, let us outline a generalized approach that may provide a theoretical understanding. Given the Algebraic Normal Form (ANF) of a Boolean function, $f:\{0,1\}^n \rightarrow \{0,1\}^m$, naturally, the maximum algebraic degree can be $n$. Thus, the reversible quantum circuit implementing $f$ can be designed with a Toffoli depth of $\ceil{\log_2 n}$, utilizing an exponential number of ancilla qubits and $\mathcal{O}(n)$ Clifford gates. This is because, in the ANF, the maximum degree is $n$, which means that $n$ different inputs are ANDed. Thus, an $n$-MCT is enough. Further, the ANF may contain at most $2^n$ terms, and that may require an exponential number of ancilla qubits. This is achieved by implementing all the required multi-controlled Toffoli (MCT) gates in parallel to compute the non-linear terms of the ANF. This follows from Theorem \ref{th:logn}.

This shows that any $n$-degree Boolean function can be implemented with a Toffoli depth of $\ceil{\log_2 n}$. This idea may also be used while considering the cryptanalytic techniques, when the quantum circuits are actually implemented. In this direction, it is understood that the depth may really be quite less for each module to be implemented. Regarding the ancilla, for practical purposes, the ANF contains poly$(n)$ many terms in it, and in such cases, the number of ancilla will also be poly$(n)$ instead of the generic exponential bound.  

\section{Conclusion}
\label{sec:con}
In this paper, we revisited the $n$-controlled Toffoli decomposition using the conditionally clean ancilla technique described by Khattar and Gidney~\cite{khattar24} and proposed an exact trade-off between Toffoli depth and the availability of clean ancilla qubits. By leveraging additional ancilla qubits, we achieved a lower Toffoli depth compared to their approach. Furthermore, we demonstrated that the conditionally clean ancillae technique cannot reduce the Toffoli depth strictly to $\ceil{\log_2 n}$, irrespective of unlimited availability of clean ancilla.

Additionally, we established that, regardless of the decomposition technique or available ancillae, the Toffoli depth of an $n$-MCT circuit is fundamentally lower-bounded by $\ceil{\log_2 n}$, with the optimal depth achieved through binary tree-based MCT decomposition. Finally, by incorporating Soeken’s measurement-based uncomputation technique (as referred in~\cite{jaques19}) for Toffoli decomposition, we extended this lower bound to T-depth as well.

\ \\
{\bf Acknowledgments:} The authors like to thank Dr. Alessandro Luongo (CQT, NUS, and Inveriant Pte. Ltd., Singapore) for his valuable inputs while preparing this paper.
Suman Dutta and Anupam Chattopadhyay acknowledge the support of `MoE AcRF Tier 1 award RT10/23'.
Subhamoy Maitra acknowledges the funding support provided by the ``Information Security Education and Awareness (ISEA) Project phase - III, Cluster - Cryptography, initiatives of MeitY, Grant No. F.No. L-14017/1/2022-HRD".

\end{document}